\documentclass[a4paper]{article}
\usepackage[latin1]{inputenc} %
\usepackage[T1]{fontenc} %
\usepackage{RR}
\usepackage{hyperref}

\usepackage[]{amsmath}
\usepackage{natbib}
\usepackage{latexsym,amssymb}
\newcommand{\C}{\ensuremath{\mathbb C}}
\newcommand{\Q}{\ensuremath{\mathbb Q}}
\newcommand{\R}{\ensuremath{\mathbb R}}
\newcommand{\Z}{\ensuremath{\mathbb Z}}
\newcommand{\K}{\ensuremath{ K}}

\newcommand{\bigO}{\mathrm{O}}
\newcommand{\mon}{\mathrm{v}}
\def\closZ{\ensuremath{\mathrm{clos}_{\mathrm{Zar}}}}

\newtheorem{Problem}{PROBLEM}[section]
\newtheorem{theorem}{Theorem}[section]

\newtheorem{proposition}[theorem]{Proposition}
\newtheorem{corollary}[theorem]{Corollary}
\newtheorem{lemma}[theorem]{Lemma}
\newtheorem{remark}[theorem]{Remark}

\newenvironment{proof}{\textit{Proof. }}{\hfill $\square$}

\RRdate{Septembre 2009}

\RRauthor{
Mohab Safey El Din\thanks{UPMC, Paris 6, LIP6}%
  \and
Lihong Zhi\thanks{KLMM, Academy of Mathematics and System Sciences, China}%
}
\authorhead{Safey El Din \& Zhi}
\RRtitle{Calcul de points rationnels dans des semi-alg\'ebriques
  convexes et d\'ecomposition en sommes de carr\'es}
\RRetitle{Computing rational points in convex semi-algebraic sets and SOS 
  decompositions}
\titlehead{Rational points in convex semi-algebraic sets}
\RRresume{Soit ${\cal P}=\{h_1, \ldots, h_s\}\subset \Z[Y_1, \ldots,
  Y_k]$, $D\geq \deg(h_i)$ pour $1\leq i \leq s$, $\sigma$ une borne
  sur la longueur binaire des coefficients des $h_i$, et $\Phi$ une
  ${\cal P}$-formule sans quantificateurs d\'efinissant un ensemble
  semi-alg\'ebrique convexe. Nous d\'ecrivons un algorithme qui
  retourne un point \`a coordonn\'ees rationnelles dans ${\cal S}$ si
  et seulement si ${\cal S}\cap \Q\neq\emptyset$. Cet algorithme est
  de complexit\'e binaire $\sigma^{\bigO(1)}D^{\bigO(k^3)}$. Si un
  point rationnel est renvoy\'e, ses coordonn\'ees sont de longueur
  binaires domin\'ees par $\sigma D^{\bigO(k^3)}$. On d\'eduit de ce
  r\'esultat une proc\'edure qui d\'ecide si un polyn\^ome $f\in
  \Z[X_1, \ldots, X_n]$ est une somme de carr\'es de polyn\^omes dans
  $\Q[X_1, \ldots, X_n]$. Soit $d$ le degr\'e de $f$, $\tau$ le
  maximum des longueurs binaires des coefficients de $f$,
  $D={{n+d}\choose{n}}$ et $k\leq D(D+1)-{{n+2d}\choose{n}}$. Cette
  proc\'edure est de complexit\'e binaire
  $\tau^{\bigO(1)}D^{\bigO(k^3)}$ et les coefficients des polyn\^omes
  obtenus en sortie ont une longueur binaire domin\'ee par $\tau
  D^{\bigO(k^3)}$.  }

\RRabstract{Let ${\cal P}=\{h_1, \ldots, h_s\}\subset \Z[Y_1, \ldots,
  Y_k]$, $D\geq \deg(h_i)$ for $1\leq i \leq s$, $\sigma$ bounding the
  bit length of the coefficients of the $h_i$'s, and $\Phi$ be a
  quantifier-free ${\cal P}$-formula defining a convex semi-algebraic
  set. We design an algorithm returning a rational point in ${\cal S}$
  if and only if ${\cal S}\cap \Q\neq\emptyset$. It requires
  $\sigma^{\bigO(1)}D^{\bigO(k^3)}$ bit operations. If a rational
  point is outputted its coordinates have bit length dominated by
  $\sigma D^{\bigO(k^3)}$. Using this result, we obtain a procedure
  deciding if a polynomial $f\in \Z[X_1, \ldots, X_n]$ is a sum of
  squares of polynomials in $\Q[X_1, \ldots, X_n]$. Denote by $d$ the
  degree of $f$, $\tau$ the maximum bit length of the coefficients in
  $f$, $D={{n+d}\choose{n}}$ and $k\leq
  D(D+1)-{{n+2d}\choose{n}}$. This procedure requires
  $\tau^{\bigO(1)}D^{\bigO(k^3)}$ bit operations and the coefficients
  of the outputted polynomials have bit length dominated by $\tau
  D^{\bigO(k^3)}$.  }
\RRmotcle{sommes de carr\'es \`a coefficients rationnels,
  programmation semi-d\'efinie positive, ensembles semi-algebraiques
  convexes, complexit\'e.}  \RRkeyword{rational sum of squares,
  semidefinite programming, convex semi-algebraic sets, complexity.}
\RRprojet{SALSA}
\RRdomaine{2} 
\RRtheme{Algorithmique, calcul certifié et cryptographie}
 \RCParis 
\RRNo{7045}
\begin{document}
\makeRR   

\section{Introduction}

\paragraph*{Motivation and problem statement.}

Suppose $f\in \R[x_1, \ldots, x_n]$,  then $f$ is a
 sum of squares (SOS) in $\R[x_1, \ldots, x_n]$ if and only if  it can be written in the
form
\begin{equation}
f=\mon^T \cdot M  \cdot \mon,
\end{equation}
in which $\mon$ is a column vector of monomials and $M$ is a real
positive semidefinite matrix \cite[Theorem 1]{PW98} (see also \cite{CLR95}). 
$M$ is also called a {\it Gram matrix} for $f$.   If $M$ has
rational entries, then $f$ is a sum of squares in $\Q[x_1, \ldots,
x_n]$.

\begin{Problem}\label{sturmfel}(Sturmfels). If $f\in \Q[x_1, \ldots, x_n]$ is a sum
of squares in $\R[x_1, \ldots, x_n]$, then is $f$  also a sum of
squares in $\Q[x_1, \ldots, x_n]$?
\end{Problem}

It 
has been pointed out  
that if there is an invertible Gram matrix for $f$, then there is a
Gram matrix for $f$ with rational entries \cite[Theorem
1.2]{Hillar08}. 
 Furthermore, 
   if $f\in \Q[x_1, \ldots, x_n]$
is a sum of $m$ squares in $\K[x_1, \ldots, x_n]$, where $\K$ is a
totally real number field with Galois closure $L$, then $f$ is also
a sum of $ 4 m \cdot 2^{[L:\Q]+1} {[L:\Q]+1\choose 2}$ squares in
$\Q[x_1, \ldots, x_n]$ \cite[Theorem 1.4]{Hillar08}. 
It is interesting to see that the number of squares can be reduced
to $m$ (see \cite{Kaltofen09}).

Although no example is known of a rational polynomial having only
irrational sum of squares, a complete answer to Question
\ref{sturmfel} is not known. This is the main motivation for us to
design an algorithm to check whether a rational polynomial having a
rational sum 
 of squares decomposition 
 and give the rational SOS 
representation if it does exist. By reducing this problem to
semi-definite programming, this can be done by designing an
algorithm checking if a convex semi-algebraic set contains rational
points (see \cite{PW98}).

\paragraph*{Main result.}We propose an algorithm which decides if a
{\em convex} semi-algebraic set ${\cal S}\subset \R^k$ contains
rational points (i.e. points with coordinates in $\Q^k$). In the
case where ${\cal S}\cap \Q^k$ is non-empty, a rational point in
${\cal S}$ is computed.

The semi-algebraic set ${\cal S}$ is given as the solution set of a
polynomial system of non-strict inequalities with integer
coefficients. Arithmetic operations, sign evaluations and
comparisons of two integers/rationals can be done in polynomial time
of the maximum bit length of the considered integers/rationals.

We bound the number of bit operations that the algorithm performs
with respect to the number of polynomials, their degrees and the
maximum bit length of their input coefficients; we also give upper
bounds on the bit length of the coordinates of the outputted
rational point if this situation occurs. More precisely, the main
result is as follows.

\begin{theorem}\label{thm:main}
  Consider a set of polynomials ${\cal P}=\{h_1, \ldots,
  h_s\}\subset 
  \Z[Y_1, \ldots, Y_k]$, and a quantifier-free ${\cal P}$-formula
  $\Phi(Y_1, \ldots, Y_k)$ and let $D$ be an integer such that
  $\deg(h_i)\leq D$ for $1\leq i\leq s$ 
  and $\sigma$ the maximum bit
  length of the coefficients of the $h_i$'s.  Let ${\cal S}\subset
  \R^k$ be the convex semi-algebraic set defined by $\Phi$.  There
  exists an algorithm which decides if ${\cal S}\cap \Q^k$ is
  non-empty within $\sigma^{\bigO(1)} (sD)^{\bigO(k^3)}$ bit
  operations. In case of non-emptiness, it returns an element of
  ${\cal S}\cap \Q^k$ whose coordinates have bit length dominated by
  $\sigma D^{\bigO(k^3)}$.
\end{theorem}

We use a procedure due to \cite{BPR96} performing quantifier
elimination over the reals in order to deduce from Theorem
\ref{thm:main} the following result.

\begin{corollary}\label{cor:main}
Let ${\cal S}\subset \R^k$ be a convex set defined by
$$ {\cal S}= \{ Y \in \R^k: 
(Q_1 X^{[1]}\in \R^{n_1})\cdots(Q_\omega X^{[\omega]}\in \R^{n_\omega})\; P(Y, X^{[1]}, \ldots, X^{[\omega]})\}$$
with quantifiers $Q_i \in \{\exists, \forall\}$, where $X^{[i]}$ is a
set of $n_i$ variables, $P$ is a Boolean function of $s$ atomic
predicates $$g(Y, X^{[1]}, \ldots, X^{[\omega]})\, \Delta_i\, 0$$
where $\Delta_i\in \{>, <, =\}$ (for $i=1, \ldots, s$) and the $g_i$'s
are polynomials of degree $D$ with integer coefficients of binary size
at most $\sigma$.  There exists an algorithm which decides if ${\cal
  S}\cap \Q^k$ is non-empty within $\sigma^{\bigO(1)}
(sD)^{\bigO(k^3\Pi_{i=1}^\omega n_i)}$ bit operations. In case of
non-emptiness, it returns an element of ${\cal S}\cap \Q^k$ whose
coordinates have bit length dominated by 
$\sigma
{D^{\bigO(k^3\Pi_{i=1}^\omega n_i)}}$.
\end{corollary}

The proof of the above results is based on quantitative and
algorithmic results for computing sampling points in semi-algebraic
sets and quantifier elimination over the reals.

It is well-known that deciding if a given polynomial $f\in \Z[X_1,
\ldots, X_n]$ of degree $d$ whose coefficients have bit length
dominated by $\tau$ is a sum of squares of polynomials in $\Q[X_1,
\ldots, X_n]$ can be reduced to a linear matrix inequality which
defines a convex semi-algebraic set (see e.g. \cite{PW98}). Applying
Theorem \ref{thm:main}, we show that there exists an algorithm
deciding if such an SOS 
decomposition exists over the rationals
and that the coefficients of the polynomials in the decomposition
have bit length dominated by $\tau D^{\bigO(k^3)}$ with
$D={{n+d}\choose{n}}$ and $k\leq D(D+1)-{{n+2d}\choose{n}}$.
Moreover, such a decomposition can be found within $\tau^{\bigO(1)}
D^{\bigO(k^3)}$ bit operations.

\paragraph*{Prior works.}

Khachiyan and Porkolab  extended the well-known result of
\cite{Lenstra83} on the polynomial-time solvability of linear
integer programming in fixed dimension  to semidefinite integer
programming. The following proposition is given  in \cite{KhPo97,
KhPo00}.
\begin{proposition}\label{KhPo97}
Let ${\cal S}\subset \R^k$ be a convex set defined as in Corollary
\ref{cor:main}.
There exists an
algorithm for solving the problem $\min\{Y_k | Y=(Y_1, \ldots, Y_k)
\in {\cal S} \bigcap \Z^k \}$ in time $\ell^{\bigO(1)}
(sD)^{\bigO(k^4)
  \Pi_{i=1}^\omega\bigO(n_i)}$. In case of non-empty, then the
minimization problem has an optimal solution whose bit length is
dominated by $ \ell D^{\bigO(k^4) \Pi_{i=1}^\omega\bigO(n_i)}$.
\end{proposition}

Their algorithm was further improved by Heinz for the case of convex
minimization where the feasible region is described by quasiconvex
polynomials \cite{Heinz05}.

Although we can apply Proposition \ref{KhPo97} directly to certify
that a given  polynomial with integer coefficients to be
non-negative for all real values of the variables by computing  a
sum of squares in $\Z[x_1, \ldots, x_n]$, the  nonnegativity of a
polynomial can be certified  if it  can be written  as a sum of
squares of polynomials in $\Q[x_1, \ldots, x_n]$. Some hybrid
symbolic-numeric algorithms have been given in \cite{PePa07, PePa08,
KLYZ08, KLYZ09} which turn a numerical sum of squares representation
of a positive polynomial into an exact rational identity.  However,
it is well known that there are plenty of polynomials which are
nonnegative but can not be written as sums of squares of
polynomials, for example, the famous Motzkin polynomial. This also
impel us to study Khachiyan and Porkolab's approach. It turns out
that by focusing on rational numbers instead of integers, we can
design an exact algorithm which decide whether a given polynomial
can be written as an SOS over the rationals and give the rational
SOS decomposition if it exists.

\paragraph*{Structure of the paper.}
Section \ref{sec:prelim} is devoted to recall the quantitative and
algorithmic results on computing sampling points in semi-algebraic
sets and quantifier elimination over the reals. Most of these results
are proved in \cite{BPR96}. Section \ref{sec:algo} is devoted to prove
the correctness of the algorithm on which Theorem \ref{thm:main} and
Corollary \ref{cor:main} rely. The complexity analysis is done in
Section \ref{sec:complexity}. In Section \ref{sec:sos}, we apply
Theorem \ref{thm:main} to prove the announced bounds on the bit length
of the rational coefficients of the decomposition into sums of squares
of a given polynomial with integer coefficients.

\paragraph*{Acknwledgments.} This work is supported by the EXACTA
grant of National Science Foundation of China (NSFC) and the French
National Research Agency (ANR). The authors thank INRIA, KLMM and the
Academy of Mathematics and System Sciences for their support.
\section{Preliminaries}\label{sec:prelim}

The algorithm on which Theorem \ref{thm:main} relies and its
complexity analysis are based on algorithmic and quantitative results
on computing sampling points in semi-algebraic sets and quantifier
elimination over the reals.

\subsection{Computing points in semi-algebraic sets}
Consider a set of polynomials ${\cal P}=\{h_1, \ldots, h_J\} \subset
\Z[Y_1, \ldots, Y_k]$, and a quantifier-free ${\cal P}$-formula
$\Phi(Y_1, \ldots, Y_k)$ (i.e. a quantifier-free formula whose atoms
is one of $h=0$, $h\neq 0$, 
 $h>0$, $h<0$ for $h\in {\cal P}$).  Let
$D$ be an integer such that $\deg(h_i)\leq D$ for $1\leq i\leq J$
and $\ell$ the maximum bit length of the coefficients of the
$h_i$'s. We denote by ${\cal S}\subset \R^k$ the semi-algebraic set
defined by $\Phi(Y_1, \ldots, Y_k)$.

A function {\sf RealizableSignConditions} computing a set of algebraic
points having a non-empty intersection with each connected component
of semi-algebraic sets defined by sign conditions satisfied by ${\cal
  P}$ is given in \cite[Section 3]{BPR96} (see also \cite[Chapter
5]{BaPoRo06}). {F}rom this, a function {\sf SamplingPoints} computing
a set of algebraic points having a non-empty intersection with each
connected component of ${\cal S}$ is obtained. These algebraic points
are encoded by
\begin{itemize}
\item a rational parametrization
  \begin{equation*}
    G=0,
    Y_1=\frac{G_1}{G_0}, \ldots,
    Y_k=\frac{G_k}{G_0}
  \end{equation*}
  where $G, G_0, \ldots, G_k$ are polynomials in $\Z[T]$ such
  that $\deg({\rm gcd}(G, G_0))=0$ and
  $$ \text{for }1\leq i \leq k, 
  \, -1\leq \deg(G_i)\leq \deg(G)-1\text{ and }0\leq
  \deg(G_0)\leq \deg(G)-1;$$ the rational parametrization is given by
  the list ${\cal G}=(G, G_0, G_1, \ldots, G_k)$; the degree of ${\cal
    G}$ is called {\em degree of the rational parametrization} and
  $Z({\cal G})\subset \C^k$ denotes the set of complex points encoded
  by ${\cal G}$;
\item and a list ${\cal T}$ of the Thom-encodings of the real roots
  $\vartheta $ of $G$ such that $\Phi\left
    (\frac{G_1(\vartheta)}{G_0(\vartheta)}, \ldots, \frac{G_k(\vartheta)}{G_0(\vartheta)}\right )$
    is true.
\end{itemize}
The bit complexity of {\sf SamplingPoints} is $\ell J^{k+1}
D^{\bigO(k)}$ and the output is such that $\deg(G)=\bigO(D)^k$ and the
bit length of the coefficients of $G, G_0, G_1, \ldots, G_k$ is
dominated by $\ell D^{\bigO(k)}$.

Factorizing over $\Q$ a univariate polynomial $h\in \Q[T]$ of degree
$\delta$ with rational coefficients of maximum bit length $\ell$ can
be done in $\ell^{\bigO(1)}\delta^{\bigO(1)}$ bit-operations (see
\cite{LLL82, Hoeij, Schonhage84}). Given a root $\vartheta$ of
$h$, the minimal polynomial of $\vartheta$ has coefficients of bit
length dominated by $\ell+\bigO(\delta)$ (see \cite{Mignotte}).

Consider now a root $\vartheta$ of $G$ and its minimal polynomial
$g$. Since $G$ and $ G_0$ are co-prime, one can compute $G_0^{-1}\mod
g$ to obtain a rational parametrization $(g, g_0, \ldots, g_k)$ with
integer coefficients of bit length dominated by $\ell D^{\bigO(k)}$
and
  $$ \text{for }1\leq i \leq k,\, 
  -1\leq \deg(g_i)\leq \deg(g)-1\text{ and }0\leq
  \deg(g_0)\leq \deg(g)-1$$ within a bit-complexity
  $\ell^{\bigO(1)}D^{\bigO(k)}$. This implies the following result.

\begin{proposition}\label{prop:resolution}
  There exists a function {\sf SemiAlgebraicSolve} which takes as
  input the system $\Phi(Y_1, \ldots, Y_k)$ and computes a rational
  parametrization ${\cal G}=(G, G_0, G_1, \ldots, G_k)$ and a list
  ${\cal T}$ of Thom-encodings such that $G$ is irreducible over $\Q$,
  and ${\cal T}$ contains the encodings of the real roots $\vartheta$
  of $G$ such that $\left (\frac{G_1(\vartheta)}{G_0(\vartheta)},
    \ldots, \frac{G_k(\vartheta)}{G_0(\vartheta)}\right )\in {\cal S}.
  $ The bit length of the coefficients of $G, G_0, G_1, \ldots, G_k$
  is dominated by $\ell D^{\bigO(k)}$ and
  $\deg(G)=\bigO(D)^k$. Moreover, {\sf SemiAlgebraicSolve} requires
  $\ell^{\bigO(1)}J^{k+1}D^{\bigO(k)}$ bit operations.
\end{proposition}

\begin{remark}\label{constantdenom}

  Since $G$ and $ G_0$ are co-prime, one can compute ${G_0}^{-1} \mod G$ 
 in polynomial time, and the binary length of its rational coefficients
  can be bounded via subresultants, we can assume, without loss of
  generality, that the rational parametrization has a constant
  denominator:
\begin{equation}\label{rationalY}
Y=\frac{1}{q} (G_1(\vartheta), G_2(\vartheta), \ldots,
G_k(\vartheta)) \in {\cal S}, ~G(\vartheta)=0,
\end{equation}
where the bit length of $q$ and the coefficients of $G, G_1, \ldots,
G_k$ are dominated by $\ell D^{\bigO(k)}$.
\end{remark}

The above discussion leads also to the following result.
\begin{proposition}\label{prop:rat0dim}
  Let ${\cal G}, {\cal T}$ be the output of ${\sf
    SemiAlgebraicSolve}(\Phi)$, $\delta$ be the degree of $G$, and
  $\ell$ be the maximum bit length of the coefficients of the
  polynomials in ${\cal G}\cup{\cal P}$.  There exists a function {\sf
    RationalZeroDimSolve} which takes as input ${\cal G}$ and $\Phi$
  and returns a rational point $y\in Z({\cal G})$ if and only if $y\in
  {\cal S}\cap Z({\cal G})\cap \Q^k$, else it returns an empty
  list. The coordinates of these rational points have bit length
  dominated by $\ell \delta^{\bigO(1)}$ and computations are performed
  within
  $\bigO(k)\bigO(J)\ell^{\bigO(1)}\delta^{\bigO(1)}{{n+D}\choose{n}}^{\bigO(1)}$
  bit operations.
\end{proposition}

\begin{remark}\label{degree1}
  According to Proposition \ref{prop:resolution}, the function {\sf
    SemiAlgebraicSolve} computes a rational parametrization ${\cal
    G}=(G, G_0, G_1, \ldots, G_k)$ such that $G$ is irreducible over
  $\Q$. Therefore a rational point $y \in Z({\cal G})$ if and only if
  $\deg(G)=1$. In order to check whether $y \in \cal S$, we only need
  to evaluate the formula $\Phi$ at $y$.
\end{remark}

  The following result is a restatement of \cite[Theorem 4.1.2]{BPR96}
  and allows us to bound the bit length of rational points in
  non-empty semi-algebraic sets defined by strict polynomial
  inequalities.
\begin{proposition}\label{prop:interior}
  Let ${\cal S}'\subset \R^k$ be a semi-algebraic set defined by a
  quantifier-free ${\cal P}$-formula whose atoms are strict
  inequalities.  Then ${\cal S}'$ contains a rational point whose
  coordinates have bit length dominated by $\ell D^{O(k)}$.
\end{proposition}

The proof of the above result (see \cite[Proof of Theorem 4.1.2
pp. 1032]{BPR96}) is based on the routine {\sf
  RealizableSignConditions} and the isolation of real roots of
univariate polynomials with rational coefficients (see
e.g. \cite[Chapter 10]{BaPoRo06}). We denote by {\sf
  RationalOpenSemiAlgebraicSolve} a function taking as input the
${\cal P}$-formula $\Phi$ and which returns a rational point in ${\cal
  S}$ if and only if there exists a non-empty semi-algebraic set
${\cal S}'$ defined by a quantifier-free ${\cal P}$-formula whose
atoms are strict inequalities such that ${\cal S}'\subset {\cal S}$.
The result below is not stated in \cite{BPR96} but is an immediate
consequence of this proof.

\begin{corollary}\label{cor:interior}
  Suppose that there exists a quantifier-free ${\cal P}$-formula whose
  atoms are strict inequalities defining a non-empty semi-algebraic
  set ${\cal S}'\subset {\cal S}$. There exists an algorithm computing
  a rational point in ${\cal S}$ if and only if ${\cal S}\neq
  \emptyset$. It requires $\ell^{\bigO(1)} J^{k+1} D^{\bigO(k)}$ bit
  operations and if a rational point is outputted, its coordinates
  have bit length dominated by $\ell D^{\bigO(k)}$.
\end{corollary}

\subsection{Quantifier elimination over the reals}
We consider now a first-order formula $F$ over the reals

$$(Q_1 X^{[1]}\in \R^{n_1})\cdots(Q_\omega X^{[\omega]}\in \R^{n_\omega})\; P(Y, X^{[1]}, \ldots, X^{[\omega]})$$
where
\begin{itemize}
\item $Y=(Y_1, \ldots, Y_k)$ is the vector of free variables;
\item each $Q_i$ ($i=1, \ldots, \omega$) is one of the quantifiers
  $\exists$ or $\forall$;
\item $P(Y, X^{[1]}, \ldots, X^{[\omega]})$ is a Boolean 
function of
  $s$ atomic predicates $$g(Y, X^{[1]}, \ldots, X^{[\omega]})\,
  \Delta_i\, 0$$ where $\Delta_i\in \{>, <, =\}$ (for $i=1, \ldots,
  s$) and the $g_i$'s are polynomials of degree $D$ with integer
  coefficients of binary size at most $\ell$.
\end{itemize}

The following result on quantifier elimination is a restatement of
\cite[Theorem 1.3.1]{BPR96}.
\begin{theorem}\label{thm:qe}
  There exists a quantified-free formula $\Psi$
$$\bigvee_{i=1}^I\bigwedge_{j=1}^{J_i} (h_{ij}\,\Delta_{ij}\, 0)$$
(where $h_{ij}\in \Z[Y_1, \ldots, Y_k]$ and $\Delta_{ij}\in \{=, >\}$)
which is equivalent to $F$ and such that
\begin{itemize}
\item $I\leq
  s^{(k+1)\Pi_{i=1}^\omega(n_i+1)}D^{(k+1)\Pi_{i=1}^\omega\bigO(n_i)}$,
\item $J_i\leq
  s^{\Pi_{i=1}^\omega(n_i+1)}D^{\Pi_{i=1}^\omega\bigO(n_i)}$, 
\item $\deg(h_{ij})\leq D^{\Pi_{i=1}^\omega\bigO(n_i)}$, 
\item the bit length of the coefficients of the polynomials $h_{ij}$
  is dominated by $\ell D^{(k+1)\Pi_{i=1}^\omega\bigO(n_i)}$.
\end{itemize}
The above transformation requires $\ell
s^{(k+1)\Pi_{i=1}^\omega(n_i+1)}D^{(k+1)\Pi_{i=1}^\omega\bigO(n_i)}$
bit operations.
\end{theorem}

In the sequel, we denote by {\sf QuantifierElimination} a function
that takes $F$ as input and returns a list $[\Psi_1, \ldots, \Psi_I]$
where the $\Psi_i's$ are the conjunctions
$$\bigwedge_{j=1}^{J_i} (h_{ij}\,\Delta_{ij}\, 0).$$
\section{Algorithm and correctness}\label{sec:algo}

\subsection{Description of the algorithm}

We use the following functions:
\begin{itemize}
\item {\sf Substitute} which takes as input a variable $Y_r\in \{Y_1,
  \ldots, Y_k\}$, a polynomial $h\in \Q[Y_1, \ldots, Y_k]$ and a
  Boolean formula $F$ and which returns a formula $\tilde{F}$ obtained
  by substituting $Y_r$ by $h$ in $F$.
\item {\sf RemoveDenominators} which takes as input a formula $F$ and
  returns a formula $\tilde{F}$ obtained by multiplying the
  polynomials in $F$ by the absolute value of the lcm of the
  denominators of their coefficients.
\end{itemize}

Consider now a rational parametrization ${\cal G}=(G, G_0, G_1,
\ldots, G_k, G_{k+1})\subset \Z[T]^{k+3}$ with $\delta=\deg(G)$. For
$0\leq i\leq \delta-1$, 
denote by $\mathbf{a}_i\in \Z^k$ the vector
of integers whose $j$-th coordinate is the coefficient of $T^i$ in
$G_j$. Similarly, for $0\leq i\leq \delta-1$, 
$\mathbf{b}_i$ denotes the coefficient of $T^i$ in $G_{k+1}$.  We
use in the sequel a function {\sf GenerateVectors} that takes as
input a rational parametrization ${\cal G}$. This function returns
the set list of couples $(\mathbf{a}_i, \mathbf{b}_i)$ for $0\leq
i\leq \delta-1$.

As in the previous section, consider now a set of polynomials ${\cal
  P}=\{h_1, \ldots, h_s\}
  \subset \Z[Y_1, \ldots, Y_k]$, and a
quantifier-free ${\cal P}$-formula $\Phi(Y_1, \ldots, Y_k)$ and let
$D$ be an integer such that $\deg(h_i)\leq D$ for $1\leq i\leq s$ 
and $\sigma$ the maximum bit length of the coefficients of the
$h_i$'s. We denote by ${\cal S}\subset \R^k$ the semi-algebraic set
defined by $\Phi(Y_1, \ldots, Y_k)$ which is supposed to be convex.

The routine {\sf FindRationalPoints} below takes as input the formula
$\Phi(Y_1, \ldots, Y_k)$ defining ${\cal S}\subset \R^k$ and the list
of variables $[Y_1, \ldots, Y_k]$.

\medskip\noindent {\sf FindRationalPoints}($\Phi, [Y_1, \ldots, Y_k]$).
\begin{enumerate}
\item\label{step:5} Let $L={\sf RationalOpenSemiAlgebraicSolve}({\sf
    Open}(\Phi))$
\item\label{step:6} If $L$ is not empty then return L
\item\label{step:1} Let ${\cal G},{\cal T}={\sf SemiAlgebraicSolve}(\Phi)$
\item\label{step:2} If ${\cal T}$ is empty then return []
\item\label{step:3} Let $L= {\sf RationalZeroDimSolve}({\cal G}, \Phi)$
\item\label{step:4} If $L$ is not empty or $k=1$ then return L
\item\label{step:7} Else
    \begin{enumerate}
    \item\label{step:a} Let $A_1, \ldots, A_k, B$ be free variables and $\Theta$ be
      the formula $$\forall Y\in \R^k\quad A_1^2+\cdots+A_k^2>0 \wedge
      (\neg \Phi\vee A_1Y_1+\cdots+A_kY_k=B)$$
    \item\label{step:b} Let $[\Psi_1, \ldots, \Psi_I]={\sf
        QuantifierElimination}(\Theta)$ and $i=1$
    \item\label{step:while} While $i\leq I$ do
      \begin{enumerate}
      \item\label{step:c} ${\cal G}, {\cal T}={\sf
          SemiAlgebraicSolve}(\Psi_i)$ and $(G, G_0, G_1, \ldots, G_k,
        G_{k+1})={\cal G}$
      \item If ${\cal T}$ is empty $i=i+1$ else break. 
      \end{enumerate}
      \item\label{step:d} Let $C={\sf GenerateVectors}(G, G_0, G_1,
        \ldots, G_{k}, G_{k+1})$
      \item\label{step:e} Let ${a}=({a}_1, \ldots, {a}_k)\neq (0,
        \ldots, 0)$ and ${b}\in \Z$ such that $({a}, {b})\in C$
      \item\label{step:f} Let $r=\max(i, 1\leq i\leq k\text{ and }{a}_i\neq 0)$
      \item\label{step:g} Let $h={b}-\frac{\sum_{j=1}^{r-1}{a}_iY_i}{{a}_r}$
      \item\label{step:h} Let $\Phi'={\sf RemoveDenominators}({\sf
          Substitute}(Y_r, h, \Phi))$
      \item\label{step:i} Let $L={\sf FindRationalPoints}(\Phi', [Y_1,
        \ldots, Y_{r-1}, Y_{r+1}, \ldots, Y_k])$
    \item\label{step:j} If $L$ is not empty, 
      \begin{enumerate}
      \item Let $(q_1, \ldots, q_{r-1}, q_{r+1}, 
      \ldots, q_k)$ be its element;
    \item \label{step:k} Let $q_r={\sf Evaluate}(\{Y_i=q_i, \,1\leq
      i\leq k,\, j\neq r\}, h)$
     \item\label{step:l} if $\Phi(q_1, \ldots,q_{r-1}, q_r,
      q_{r+1},\ldots , q_k)$ is true, return $[(q_1, \ldots, q_{r-1},
      q_r, q_{r+1},\ldots, q_k)]$ else return [].
     \end{enumerate}
    \item Else return []. 
    \end{enumerate}

\end{enumerate}

\begin{proposition}
  The algorithm {\sf FindRationalPoints} returns a list containing a
  rational point if and only if ${\cal S}\cap \Q^k$ is non-empty, 
  else it returns an empty list.
\end{proposition}

The next paragraph is devoted to prove this proposition.
\begin{remark}
Let ${\cal S}\subset \R^k$ be a convex set defined by
$$ {\cal S}= \{ Y \in \R^k: \R^k (Q_1 X^{[1]}\in \R^{n_1})\cdots(Q_\omega X^{[\omega]}\in \R^{n_\omega})\; P(Y, X^{[1]}, \ldots, X^{[\omega]})\}$$
with quantifiers $Q_i \in \{\exists, \forall\}$, where $X^{[i]}$ is a
set of $n_i$ variables, $P$ is a Boolean function of $s$ atomic
predicates $$g(Y, X^{[1]}, \ldots, X^{[\omega]})\, \Delta_i\, 0$$
where $\Delta_i\in \{>, <, =\}$ (for $i=1, \ldots, s$).

Denote by $\Theta$ the quantified formula defining ${\cal S}$ and by
$[\Psi_1, \ldots, \Psi_I]$ the output of ${\sf
  QuantifierElimination}(\Theta)$. Running {\sf FindRationalPoints} on
the $\Psi_i$'s allows to decide the existence of rational points in
${\cal S}$. This proves a part of Corollary \ref{cor:main}.
\end{remark}
\subsection{Proof of correctness}

In the sequel, we denote by $\closZ({\cal S})$ its
Zariski-closure. Following \cite[Definition 2.8.1 and Proposition
2.8.2 pp. 50]{BCR}, we define the {\em dimension} of ${\cal S}$ as the
Krull dimension of the ideal associated to $\closZ({\cal S})$. By
convention, the dimension of the empty set is $-1$.

We reuse the notations introduced in the description of {\sf
  FindRationalPoints}. The proof is done by induction on $k$. Before
investigating the case $k=1$, we recall some elementary facts.

\subsection*{Preliminaries.} We start with a lemma.
\begin{lemma}\label{lemma:interior}
  Let $A\subset \R^k$ be a semi-algebraic set defined by a
  quantifier-free ${\cal P}$-formula. If $\dim(A)=k$ there exists
  $y\in \R^k$ such that for all $h\in {\cal P}$ $h(y)>0$ or $h(y)<0$.
\end{lemma}
\begin{proof}
  Suppose that for all $y\in A$, there exists $h\in {\cal P}$ such
  that $h(y)=0$. Then, $A$ is contained in the union ${\cal H}$ of the
  hypersurfaces defined by $h=0 $ for $h\in {\cal P}$. Consequently,
  $\dim(A)\leq \dim({\cal H})<k$, which contradicts $\dim(A)=k$.
\end{proof}

The following lemma recalls an elementary property of convex
semi-algebraic sets of dimension $0$.

\begin{lemma}\label{lemma:point}
  Let $A\subset \R^k$ be a convex semi-algebraic set. If $\dim(A)=0$,
  then $A$ is reduced to a single point.
\end{lemma}

\begin{proof}
  If there exist two distinct points $y_1, y_2$ in $A$, the set
  $B=\{ty_1+(1-t)y_2, t\in [0,1]\}$ is contained in $A$. This implies
  that $\closZ(B)\subset \closZ(A)$ and consequently $\dim(B)\leq
  \dim(A)$. Since $\closZ(B)$ is the line containing $y_1$ and $y_2$,
  $\dim(B)=1$ and $\dim(A)\geq 1$ which contradicts the assumption
  $\dim(A)=0$. Our claim follows.
\end{proof}

\subsection*{Correctness when $k=1$.}

\begin{lemma}\label{lemma:k1}
  Suppose that $k=1$. Then Steps (\ref{step:5}-\ref{step:4}) return a rational
  point in ${\cal S}$ if and only if ${\cal S}\cap \Q^k\neq \emptyset$
  else an empty list is returned.
\end{lemma}

\begin{proof}
If $k=1$, the dimension of ${\cal S}$ is either $1$, 
$-1$ or $0$.
\begin{enumerate}
\item Suppose that ${\cal S}$ has dimension $1$. {F}rom Lemma
  \ref{lemma:interior}, there exists a non-empty semi-algebraic set
  ${\cal S}'\subset {\cal S}$ defined by a quantifier-free ${\cal
    P}$-formula whose atoms are strict inequalities. Thus ${\cal S}'$
  contains a rational point. {F}rom Corollary \ref{cor:interior}, such
  a rational point in ${\cal S}$ is outputted at Step (\ref{step:5}).
\item Suppose that ${\cal S}$ has dimension $-1$ (i.e. ${\cal S}$ is
  empty). {F}rom Proposition \ref{prop:resolution}, the list of Thom-encodings
  outputted at Step (\ref{step:1}) is empty and the empty list is
  returned at Step (\ref{step:2}).
\item Suppose that ${\cal S}$ has dimension $0$. {F}rom Lemma
  \ref{lemma:point}, ${\cal S}$ is a single point contained in
  $Z({\cal G})$. {F}rom Proposition \ref{prop:rat0dim}, this point is
  outputted at Step (\ref{step:3}) if and only if it is a rational
  point; else the empty list is outputted.
\end{enumerate}
\end{proof}

\subsection*{The case $k>1$.}

Our induction assumption is that, given a quantifier-free ${\cal
  P}'$-formula $\Phi'$ (with ${\cal P}'\subset \Z[Y_1, \ldots,
Y_{k-1}]$) defining a convex semi-algebraic set ${\cal S}'\subset
\R^{k-1}$, {\sf FindRationalPoints} returns a list containing a
rational point if and only if ${\cal S}'\cap \Q^{k-1}$ is non-empty,
else it returns an empty list.



\begin{lemma}\label{lemma:crucial}
  Suppose that $0\leq \dim({\cal S})<k$. There exists $(a_1,\ldots,
  a_k)\in \R^k$ and $b\in \R$ such that $(a_1, \ldots,
  a_k)\neq (0, \ldots, 0)$ and
  \begin{equation}
    \label{eq:3}
    \forall (y_1,\ldots, y_k) \in
    \R^k\quad (y_1, \ldots, y_k)\in {\cal S}\Longrightarrow
    a_1y_1+\cdots+a_ky_k=b. 
  \end{equation}
\end{lemma}

\begin{proof}
  It is sufficient to prove that $\closZ({\cal S})$ is an affine
  subspace over $\R$: in this case, there exists a real affine
  hyperplane $H$ (defined by $\sum_{i=1}^k a_iY_i=b$ for $(a_1,
  \ldots, a_k)\in \R^k\setminus (0, \ldots, 0)$ and 
  $b\in \R$) such that ${\cal S}\subset \closZ({\cal S})\subset H$.

  We prove below that $\closZ({\cal S})\cap\R^k$ is an affine subspace
  which implies that $\closZ({\cal S})$ is an affine subspace.

  {F}rom Lemma \ref{lemma:point}, if $\dim({\cal S})=0$ then ${\cal
    S}$ is a single point; thus the conclusion follows immediately.

  We suppose now that $\dim({\cal S})>0$; hence ${\cal S}$ is not
  empty and contains infinitely many points.  Consider $y_0\in {\cal
    S}$. Given $y\in \R^k\setminus \{y_0\}$, we denote by
  $L_{y_0,y}\subset \R^k$ the real line containing $y$ and $y_0$ and
  by $H_{y_0, y}\subset \R^k$ the real affine hyperplane which is
  orthogonal to $L_{y_0,y}$ and which contains $y_0$.

  Since ${\cal S}$ is convex, for all $y\in {\cal S}\setminus
  \{y_0\}$, ${\cal S}\cap L_{y,y_0}\neq\emptyset$. We consider the set
  ${\cal U}_{y_0}=\bigcap_{y\in {\cal S}\setminus\{y_0\}} H_{y_0,y}$;
  note that ${\cal U}_{y_0}$ is an affine subspace since it is the
  intersection of affine subspaces. We claim that the orthogonal of
  ${\cal U}_{y_0}$ is $\closZ({\cal S})\cap\R^k$.

  We first prove that ${\cal S}$ is contained in the orthogonal of
  ${\cal U}_{y_0}$ which implies that $\closZ({\cal S})\cap\R^k$ is
  contained in the orthogonal of ${\cal U}_{y_0}$.  By definition of
  ${\cal U}_{y_0}$, for all $u\in {\cal U}_{y_0}$ and all $y\in {\cal
    S}\setminus\{y_0\}$, the inner product of $\overrightarrow{y_0 u}$
  and $\overrightarrow{y_0, y}$ is zero.  We prove now that the
  orthogonal of ${\cal U}_{y_0}$ is contained in $\closZ({\cal
    S})\cap\R^k$. By definition, the orthogonal of ${\cal U}_{y_0}$ is
  the set of lines $L_{y,y_0}$ for $y\in {\cal S}\setminus
  \{y_0\}$. Thus, it is sufficient to prove that for all $y\in {\cal
    S}\setminus \{y_0\}$, $L_{y,y_0}$ is contained in $\closZ({\cal
    S})\cap\R^k$. For all $y\in {\cal S}\setminus \{y_0\}$, ${\cal
    S}\cap L_{y,y_0}\neq \emptyset$ because ${\cal S}$ is
  convex. Moreover, $\closZ({\cal S}\cap L_{y,y_0})\cap\R^k$ is
  $L_{y,y_0}$. Since ${\cal S}\cap L_{y,y_0}\subset {\cal S}$,
  $L_{y,y_0}$ is contained in $\closZ({\cal S})\cap\R^k$.  Our
  assertion follows.
\end{proof}

Suppose that $\dim({\cal S})=k$. Then, by Lemma \ref{lemma:interior},
${\cal S}\cap \Q^k$ is not empty and a rational point is outputted at
Step (\ref{step:6}) by Corollary \ref{cor:interior}. Suppose now that
${\cal S}$ is empty. Then, an empty list is returned at Step
(\ref{step:2}). We suppose now that ${\cal S}$ is not empty and that
no rational point is outputted at Step (\ref{step:4}). Hence, we enter
at Step (\ref{step:7}).

Remark that the formula $\Theta$ (Step (\ref{step:a})) defines the
semi-algebraic set ${\cal A}\subset\R^k\times \R$ such that
$(a_1, \ldots, a_k, b)\in {\cal A}$ if and only if
$(a_1, \ldots, a_k)\neq (0, \ldots, 0)$ and
$$\forall (y_1,\ldots, y_k) \in \R^k\quad (y_1, \ldots,
y_k)\in {\cal S}\Longrightarrow a_1y_1+\cdots+a_ky_k=b.$$ Thus, the
quantifier-free formula $\bigvee_{i=1}^I\Psi_i$ (Step
(\ref{step:b})) 
defines ${\cal A}$. Note that by Lemma \ref{lemma:crucial}, ${\cal
A}$ is not empty. Hence, the loop at Step (\ref{step:while}) ends by
finding a rational parametrization ${\cal G}=(G, G_0, G_1, \ldots,
G_k, G_{k+1})$ (computed at Step (\ref{step:c})) which encodes some
points in ${\cal A}$.

{F}rom the specification of {\sf SemiAlgebraicSolve}, $G$ is
irreducible over $\Q$. Let $a=(a_1, \ldots, a_k)\in \R^k$ and $b\in
\R$ such that $(a, b)\in {\cal A}\cap Z({\cal G})$. Then, there exists
a real root $\vartheta$ of $G$ such that
\begin{equation}
  \label{eq:1}
G_0(\vartheta)\left (\begin{array}{c}
    a\\
    b
\end{array}\right )=
\sum_{i=1}^{\deg(G)-1}\vartheta^i
\left (\begin{array}{c}
\mathbf{a}_i\\
\mathbf{b}_i
\end{array}\right )
\end{equation}
where the couples $(\mathbf{a}_i, \mathbf{b}_i)\in \Z^k\times\Z$ are
those returned by {\sf GenerateVectors} (Step (\ref{step:d})). Since
${\rm gcd}(G_0, G)=1$, $G_0(\vartheta)\neq 0$. Moreover, $(a,b)\in
{\cal A}$ implies $a\neq (0, \ldots, 0)$. Note also that $(a, b)\in
{\cal A}$ implies that for all $\lambda\in \R^\star$, $(\lambda
a,\lambda b)\in {\cal A}$ since for all $(y_1, \ldots, y_k)\in {\cal
  S}$ and $\lambda\in \R^\star$
$$a_1y_1+\cdots+a_ky_k=b\Longleftrightarrow
\lambda(a_1y_1+\cdots+a_ky_k)=\lambda b$$ This proves
that $$(a^\star, b^\star)=(G_0(\vartheta)a,
G_0(\vartheta)b)\in {\cal A}\text{ and }(a^\star_1,
\ldots,a^\star_k)\neq (0, \ldots, 0).$$ Thus, there exists $i$
such that $\mathbf{a}_i\neq 0$, which implies that Step (\ref{step:e})
never fails. To end the proof of correctness, we distinguish the case
where ${\cal S}\cap \Q^k$ is empty or not.

\paragraph*{The non-empty case.}
We suppose first that ${\cal S}\cap \Q^k$ is non-empty; let $(y_1,
\ldots, y_k)\in {\cal S}\cap \Q^k$. Using (\ref{eq:1}), the linear
relation $a^\star_1y_1+\cdots+a^\star_k y_k=b^\star$ 
implies the algebraic relation of degree $\deg(G)-1$:
\begin{equation}
  \label{eq:2}
\sum_{i=0}^{\deg(G)-1}\vartheta^i(\sum_{j=1}^k\mathbf{a}_{i,j}y_j-\mathbf{b}_i)=0, 
\end{equation}
where $\mathbf{a}_{i,j}$ 
is the $j$-th coordinate of $\mathbf{a}_i$. 
Since $G$ is irreducible, it is the minimal polynomial of
$\vartheta$; hence $\vartheta$ is an algebraic number of degree
$\deg(G)$.  Thus, (\ref{eq:2}) is equivalent to
$$\forall 0\leq i \leq \deg(G)-1, \quad
\sum_{j=1}^k\mathbf{a}_{i,j}y_j=\mathbf{b}_i. 
$$ 
We previously proved that there exists $i$ such that $\mathbf{a}_i\neq
0$. We let $a=(a_1, \ldots, a_k)\in \Z^k\setminus (0, \ldots, 0)$ and
$b\in \Z$ be respectively the vector with integer coordinates and the
integer chosen in $C$ (Step (\ref{step:e})). We have just proved that
${\cal S}\cap \Q^k$ is contained in the intersection of ${\cal S}$ and
of the affine hyperplane $H$ defined by $a_1
Y_1+\cdots+a_kY_k=b$. Note also that ${\cal S}\cap H$ is convex since
${\cal S}$ is convex and $H$ is an affine hyperplane.

Consider the projection $\pi_r: (y_1, \ldots, y_k)\in \R^k\rightarrow
(y_1, \ldots, y_{r-1}, y_{r+1}, \ldots, y_k)\in
\R^{k-1}$ 
for the integer $r$ computed at Step (\ref{step:f}). It is clear that
the formula $\Phi'$ computed at Step (\ref{step:h}) defines the
semi-algebraic set $\pi_r({\cal S}\cap H)\subset \R^{k-1}$. Since
${\cal S}\cap H$ is convex, $\pi_r({\cal S}\cap H)$ is convex. Thus,
the call to {\sf FindRationalPoints} (Step (\ref{step:i})) with inputs
$\Phi'$ and $[Y_1, \ldots, Y_{r-1}, Y_{r+1}, \ldots, Y_k]$ is
valid. {F}rom the induction assumption, it returns a rational point in
$\pi_r({\cal S}\cap H)$ if and only if $\pi_r({\cal S}\cap H)$ has a
non-empty intersection with $\Q^{k-1}$.

Since ${\cal S}\cap \Q^k$ (which is supposed to be non-empty) is
contained in ${\cal S}\cap H$, $\pi_r({\cal S}\cap H)$ contains
rational points. Thus, the list $L$ (Step (\ref{step:i})) contains a
rational point $\mathbf{q}_{k-1}=(q_1, \ldots, q_{r-1}, q_{r+1},
\ldots, q_k)\in \pi_r({\cal S}\cap H)$. This implies that
$\pi_{r}^{-1}(\mathbf{q}_{k-1})\cap H$ has a non-empty intersection
with ${\cal S}\cap H$. Remark that $\pi_{r}^{-1}(\mathbf{q}_{k-1})\cap
H$ is the rational point $\mathbf{q}=(q_1, \ldots, q_{r-1}, q_r,
q_{r+1}, \ldots, q_k)$ where $q_r$ is computed at Step
(\ref{step:k}). It belongs to ${\cal S}$ since
$\pi_{r}^{-1}(\mathbf{q}_{k-1})\cap H$ and ${\cal S}\cap H$ have a
non-empty intersection. Thus, $\Phi(q_1, \ldots, q_{r-1}, q_r,
q_{r+1}, \ldots, q_k)$ is true and $\mathbf{q}$ is returned by {\sf
  FindRationalPoints}.

\paragraph*{The empty case.} Suppose now that ${\cal S}\cap \Q^k$ is
empty.  As above $H$ denotes the affine hyperplane defined by
$a_1Y_1+\cdots+a_kY_k=b$ where $(a_1, \ldots, a_k)\in \Z^k$ and
$b\in \Z$ are chosen at Step (\ref{step:e}). Using the above
argumentation, $\pi_r({\cal S}\cap H)$ is convex and the formula
$\Phi'$ (Step (\ref{step:h})) defines $\pi_r({\cal S}\cap H)$. Thus,
the call to {\sf FindRationalPoints} (Step (\ref{step:i})) with
inputs $\Phi'$ and $[Y_1, \ldots, Y_{r-1}, Y_{r+1}, \ldots, Y_k]$ is
valid. Suppose that $\pi_r({\cal S}\cap H)$ does not contain
rational points. Then, by the induction assumption, $L$ is empty and
the empty list is returned (Step (\ref{step:j})) which is the
expected output since we have supposed 
 ${\cal S}\cap
\Q^k=\emptyset$. Else, $L$ contains a rational point $(q_1, \ldots,
q_{r-1}, q_{r+1}, \ldots, q_k)$. Consider the rational point $(q_1,
\ldots, q_{r-1}, q_r, q_{r+1}, \ldots, q_k)$ (where $q_r$ is
computed at Step (\ref{step:k})). It can not belong to ${\cal S}$
since we have
supposed ${\cal S}\cap \Q^k$ is empty. 
Consequently, $\Phi(q_1, \ldots,
q_{r-1}, q_r, q_{r+1}, \ldots, q_k)$ is false and the empty list is
returned.

\section{Complexity}\label{sec:complexity}

We analyze now the bit complexity of {\sf FindRationalPoints}.

\begin{proposition}\label{prop:complexite}
  Consider a set of polynomials ${\cal P}=\{h_1, \ldots, h_s\}
  \subset
  \Z[Y_1, \ldots, Y_k]$, and a quantifier-free ${\cal P}$-formula
  $\Phi(Y_1, \ldots, Y_k)$ and let $D$ be an integer such that
  $\deg(h_i)\leq D$ for $1\leq i\leq s$ 
   and $\sigma$ the maximum bit
  length of the coefficients of the $h_i$'s.  Then, ${\sf
    FindRationalPoints}(\Phi, [Y_1, \ldots, Y_k])$ requires
  $\sigma^{\bigO(1)} (sD)^{\bigO(k^3)}$ bit operations. Moreover, if
  it outputs a rational point, its coordinates have bit length
  dominated by $\sigma D^{\bigO(k^3)}$.
\end{proposition}

\begin{remark}
Let ${\cal S}\subset \R^k$ be a convex set defined by
$$ {\cal S}= \{ Y \in \R^k: \R^k (Q_1 X^{[1]}\in \R^{n_1})\cdots(Q_\omega X^{[\omega]}\in \R^{n_\omega})\; P(Y, X^{[1]}, \ldots, X^{[\omega]})\}$$
with quantifiers $Q_i \in \{\exists, \forall\}$, where $X^{[i]}$ is a
set of $n_i$ variables, $P$ is a Boolean function of $s$ atomic
predicates $$g(Y, X^{[1]}, \ldots, X^{[\omega]})\, \Delta_i\, 0$$
where $\Delta_i\in \{>, <, =\}$ (for $i=1, \ldots, s$) and the $g_i$'s
are polynomials of degree $D$ with integer coefficients of binary size
at most $\sigma$. Denote by $\Theta$ the quantified formula defining
${\cal S}$. By Theorem \ref{thm:qe}, ${\sf
  QuantifierElimination}(\Theta)$ requires $\sigma
s^{(k+1)\Pi_{i=1}^\omega(n_i+1)}D^{(k+1)\Pi_{i=1}^\omega\bigO(n_i)}$
bit operations.

It outputs a list of conjunctions 
 $\Phi_1, \ldots,
\Phi_I$ with $I\leq
s^{(k+1)\Pi_{i=1}^\omega(n_i+1)}D^{(k+1)\Pi_{i=1}^\omega\bigO(n_i)}$,
and for $1\leq i\leq I$, $\Phi_i $ is a conjunction 
of $J_i\leq
s^{\Pi_{i=1}^\omega(n_i+1)}D^{\Pi_{i=1}^\omega\bigO(n_i)}$ atomic
predicates $h\,\Delta\, 0$ with $h\in \Z[Y_1, \ldots, Y_k]$,
$\Delta\in \{=, >\}$ and $\deg(h)\leq
D^{\Pi_{i=1}^\omega\bigO(n_i)}$ and the bit length of the
coefficients of the polynomials $h_{ij}$ is dominated by $\sigma
D^{(k+1)\Pi_{i=1}^\omega\bigO(n_i)}$. Thus, the cost of running {\sf
FindRationalPoints} on all the $\Phi_i$'s requires
$\sigma^{\bigO(1)} (sD)^{\bigO(k^3\Pi_{i=1}^\omega n_i)}$ bit
operations. In case of non-emptiness of ${\cal S}\cap\Q^k$, it
returns an element of ${\cal S}\cap \Q^k$ whose coordinates have bit
length dominated by $\sigma {D^{\bigO(k^3\Pi_{i=1}^\omega n_i)}}$.
This ends to prove Corollary \ref{cor:main}.
\end{remark}

We start with a lemma.
\begin{lemma}\label{lemma:initialisation}
  Steps (\ref{step:5}-\ref{step:4}) of ${\sf
    FindRationalPoints}(\Phi)$ perform within $\sigma^{\bigO(1)}
  s^{k+1} D^{\bigO(k)}$ bit operations. If a rational point is
  returned at Step (\ref{step:4}) or Step (\ref{step:6}), its
  coordinates have bit length dominated by $\sigma D^{\bigO(k)}$.
\end{lemma}

\begin{proof}
  The result is a direct consequence of the results stated at Section
  \ref{sec:prelim}.
\begin{enumerate}
\item {F}rom Corollary \ref{cor:interior}, Step (\ref{step:5}) is
  performed within $\sigma s^{k+1}D^{\bigO(k)}$ bit operations and if a
  rational point is outputted at Step (\ref{step:6}), its coordinates
  have bit length dominated by $\sigma D^{\bigO(k)}$.
\item {F}rom Proposition \ref{prop:resolution}, Steps (\ref{step:1})
  and (\ref{step:2}) are performed within
  $\sigma^{\bigO(1)}s^{k+1}D^{\bigO(k)}$ bit operations.
\item {F}rom Proposition \ref{prop:rat0dim}, Step (\ref{step:3})
  requires $\sigma^{\bigO(1)}D^{\bigO(k)}$ bit operations. Moreover,
  if a rational point is outputted at Step (\ref{step:4}), its
  coordinates have bit length dominated by $\sigma D^{\bigO(k)}$.
\end{enumerate}
\end{proof}

We prove now the following result.
\begin{lemma}\label{lemma:intermediate}
  \begin{enumerate}
  \item\label{assert:1} Steps (\ref{step:a}-\ref{step:h}) require
    $\sigma^{\bigO(1)} (sD)^{\bigO(k^2)}$ bit operations. The number
    of polynomials in $\Phi'$ is $s$; their degrees are dominated by
    $D$ and the bit length of their coefficients is dominated by
    $\sigma D^{\bigO(k^2)}$.
  \item\label{assert:2} If a rational point with coordinates of bit
    length dominated by $\ell$ is returned at Steps
    (\ref{step:i}-\ref{step:j}), the rational number computed at Step
    (\ref{step:k}) has bit length dominated by $\ell+\sigma
    D^{\bigO(k^2)}$.
  \end{enumerate}
\end{lemma}

\begin{proof} {F}rom Theorem \ref{thm:qe}, Steps
  (\ref{step:a}-\ref{step:b}) are performed within $\sigma
  s^{\bigO(k^2)}D^{\bigO(k^2)}$ bit operations. The obtained
  quantifier-free formula is a disjunction of $(sD)^{\bigO(k^2)}$
  conjunctions. 
  Thus the loop (Step (\ref{step:while})) 
  makes at
  most $(sD)^{\bigO(k^2)}$ calls to {\sf SemiAlgebraicSolve}.  Each
   conjunctions 
   involves $(sD)^{\bigO(k)}$ polynomials of degree
  $D^{\bigO(k)}$ in $\Z[A_1, \ldots, A_k, B]$ with integers of bit
  length dominated by $\sigma D^{\bigO(k^2)}$.

  Thus, {f}rom Proposition \ref{prop:resolution}, Step (\ref{step:c})
  is performed within $\sigma^{\bigO(1)}(sD)^{\bigO(k^2)}$ bit
  operations and outputs a rational parametrization of degree
  $D^{\bigO(k^2)}$ with integer coefficients of bit length dominated
  by $\sigma D^{\bigO(k^2)}$. Thus, the integers in the list computed
  at Step (\ref{step:d}) have bit length dominated by $\sigma
  D^{\bigO(k^2)}$. This implies that the polynomial obtained from
  Steps (\ref{step:e}-\ref{step:g}) has rational coefficients of bit
  length dominated by $\sigma D^{\bigO(k^2)}$. Assertion
  (\ref{assert:2}) follows immediately.

  The bit complexity of these steps is obviously negligible compared
  to the cost of Step (\ref{step:c}). The substitution phase (Step
  \ref{step:h}) has a cost which is still dominated by the cost of
  Step (\ref{step:c}). As announced, the obtained formula $\Phi'$
  contains $s$ $(k-1)$-variate polynomials of degree $D$ with integer
  coefficients of bit length dominated by $\sigma D^{\bigO(k^2)}$.
\end{proof}

We prove now Proposition \ref{prop:complexite} by induction on $k$.
The initialization of the induction is immediate {f}rom Lemmata
\ref{lemma:k1} and \ref{lemma:initialisation}.

Suppose that $k>1$. Suppose that the execution of ${\sf
  FindRationalPoints}(\Phi)$ stops at Steps (\ref{step:6}), or
(\ref{step:2}) or (\ref{step:4}).  {F}rom Lemma
\ref{lemma:initialisation}, we are done. Suppose now that we enter in
Step (\ref{step:7}).

By Lemma \ref{lemma:intermediate}(\ref{assert:1}), the formula $\Phi'$
computed at Step (\ref{step:h}) contains $s$ $(k-1)$-variate
polynomials of degree $D$ and coefficients of bit length dominated by
$\sigma D^{\bigO(k^2)}$ and is obtained within $\sigma^{\bigO(1)}
(sD)^{\bigO(k^2)}$ bit operations. The induction assumption implies
that
\begin{itemize}
\item Step (\ref{step:i}) requires
  $\sigma^{\bigO(1)}(sD)^{\bigO(k^3)}$ bit operations, 
\item If a rational point is contained in $L$ (Step (\ref{step:j})),
  its coordinates have bit length dominated by $\sigma
  D^{\bigO(k^3)}$.
\end{itemize}
Hence, by Lemma \ref{lemma:intermediate}(\ref{assert:2}), the rational
number computed at Step (\ref{step:k}) has bit length dominated by
$\sigma D^{\bigO(k^3)}$. Moreover, the cost of Steps
(\ref{step:k}-\ref{step:l}) is negligible compared to the cost of
previous steps.

\section{Rational sums of squares}\label{sec:sos}

Consider a polynomial $f\in \Z[x_1, \ldots, x_n]$ of degree $2d$
whose coefficients have bit length bounded by $\tau$. If we choose
$\mon$ as the vector of all monomials in $\Z[x_1, \ldots, x_n]$ of
degree less than or equal to
$d$, 
then we consider the set of real symmetric matrices $M=M^T$ of
dimension $D={{n+d}\choose{n}}$ for which $f=\mon^T \cdot M \cdot
\mon $.  By Gaussian elimination, it follows that there exists an
integer $k\leq \frac{1}{2}D(D+1)-{{n+2d}\choose{n}}$ such that
\begin{equation}
M=\{ M_0+ Y_1 M_1+\ldots+Y_k M_k,  ~ Y_1, \ldots, Y_k \in \R \}
\end{equation}
for some rational symmetric matrices $M_0, \ldots, M_k$.  The
polynomial $f$ can be written as 
 a sum of squares of polynomials 
 if and only if the
matrix $M$ can be completed as a symmetric positive semidefinte
matrix (see \cite{Monique01}).
 Let
$Y=(Y_1, \ldots, Y_k)$, we define
\begin{equation}\label{Yset}
{\cal S}=\{Y \in \R^k ~|~ M(Y) \succeq 0,  ~M(Y)=M(Y)^T, ~f=\mon^T
\cdot M(Y) \cdot \mon \}.
\end{equation}
It is clear that ${\cal S} \subseteq \R^k$ is a convex set defined
by setting all polynomials in
\begin{equation}
\Phi(Y_1, \ldots, Y_k) =\{(-1)^{(i+D)} m_i, ~i=0, \ldots, D-1\}
\end{equation}
to be nonnegative, where the $m_i$'s are the coefficients of the
characteristic polynomial of $M(Y)$. The cardinality $s$ of ${\Phi}$
is bounded by $D$ and ${\Phi}$ 
contains polynomials of degree
bounded by $D$ whose coefficients have bit length bounded by $\tau
D$ (see \cite{PW98}). Hence the semi-algebraic set defined by
(\ref{Yset}) is
\begin{equation}\label{Ysetm}
{\cal S} = \{ (Y_1, \ldots, Y_k) \in \R^k ~|~ (-1)^{(i+D)} m_i \geq
0, ~ 0 \leq i \leq D-1 \}.
\end{equation}

The result below is obtained by applying  Theorem \ref{thm:main} to
the semi-algebraic set defined above.

\begin{corollary}
  Let $f\in \Z[x_1, \ldots, x_n]$ of degree $2d$ 
  with integers of bit
  length bounded by $\tau$.  By running the algorithm {\sf
    FindRationalPoints} for the semi-algebraic set defined in
  (\ref{Yset}), one can decide whether $f$ is a sum of squares in
  $\Q[x_1, \ldots, x_n]$ within $\tau^{\bigO(1)} D^{\bigO(k^3)}$ bit
  operations. Suppose $f=\sum f_i^2, f_i\in \Q[x_1, \ldots, x_n]$,
  then the bit lengths of rational coefficients of the
  $f_i$'s are bounded
  by $\tau D^{\bigO(k^3)}$.
\end{corollary}

\begin{remark}
  Applying Proposition \ref{KhPo97} by Khachiyan and Porkolab to the
  semi-algebraic set defined in (\ref{Yset}), one can decide whether
  $f$ is a sum of squares in $\Z[x_1, \ldots, x_n]$ within
  $\tau^{\bigO(1)} D^{\bigO(k^4)}$ operations. Suppose that $f=\sum
  f_i^2, f_i\in \Z[x_1, \ldots, x_n]$, then the bit lengths of integer
  coefficients of $f_i$ are bounded by $\tau D^{\bigO(k^4)}$.
\end{remark}

Porkolab and Khachiyan showed that the non-emptiness of the convex set
defined in (\ref{Yset}) over the reals can be determined in $\bigO(k
D^4)+ D^{\bigO(\min\{k, D^2\})}$ arithmetic operations over $\ell
D^{\bigO(\min\{k, D^2\})}$-bit numbers, where $\ell$ is the maximal
bit length of the matrices $M_i$ (see \cite{PoKh97}).  Suppose ${\cal
  S} \neq \emptyset$, i.e., $f\in \Q[x_1, \ldots, x_n]$ is a sum of
$m$ squares in $\K[x_1, \ldots, x_n]$ where $\K$ is an algebraic
extension of $\Q$.  If $\K$ is a totally real number field, then $f$
is also a sum of squares in $\Q[x_1, \ldots, x_n]$, i.e, ${\cal S}
\bigcap \Q^n \neq \emptyset$ 
 (see \cite{Hillar08,Kaltofen09}).  The following
lemma and proof can be deduced from arguments given in
\cite{Kaltofen09}.

\begin{lemma}
Suppose ${\cal G}=(G, G_0, G_1, \ldots, G_k)$ is  a rational
parametrization   for the semi-algebraic set $\cal S$ defined in
(\ref{Yset}) computed  by {\sf SemiAlgebraicSolve}. Suppose
$\vartheta$ is a real root of $G$ such that
\begin{equation}\label{rationalY}
Y(\vartheta)=\frac{1}{q} (G_1(\vartheta), G_2(\vartheta), \ldots,
G_k(\vartheta)) \in {\cal S},
\end{equation}
Then for any real root $\vartheta_i$ of $G$,  we have
\begin{equation}\label{rationalTY}
Y(\vartheta_i)=\frac{1}{q} (G_1(\vartheta_i), G_2(\vartheta_i),
\ldots, G_k(\vartheta_i)) \in {\cal S}.
\end{equation}
 Moreover, if the
polynomial $G$ has only real roots, then the point defined by
$\frac{1}{\deg{G}} \sum_{i=1}^{\deg{G}}  Y(\vartheta_i)$ is a
rational point in $\cal S$.
\end{lemma}

\begin{proof}
Since $Y(\vartheta) \in {\cal S}$, the matrix $M(Y(\vartheta))$ is
positive semidefinite. We can perform the Gaussian elimination over
$\Q(\vartheta)$ to obtain the decomposition
$M(Y(\vartheta))=A(\vartheta)^T A(\vartheta)$.  It is clear  that
for any real root $\vartheta_i$ of $G$,
$M(Y(\vartheta_i))=A(\vartheta_i)^T A(\vartheta_i)$ is also positive
semi-definite, i.e., $Y(\vartheta_i) \in {\cal S}$. Moreover, if $G$
has only real roots $\vartheta_i$, then $\sum_{\vartheta_i,
G(\vartheta_i)=0}G_j(\vartheta_i) \in \Q$.  It follows that the
point defined by $\frac{1}{\deg{G}} \sum_{i=1}^{\deg{G}}
Y(\vartheta_i)$ is a rational point in $\cal S$.

\end{proof}

The above discussion leads to the following result.

\begin{theorem}
Suppose $f\in \Z[x_1, \ldots, x_n]$.   There exists a function {\sf
RationalTotalRealSolve} which either determines that $f$ can not be
written as sum of squares over the reals  or returns a sum of
squares representation of $f$ over $\Q[x_1, \ldots, x_n]$ if and
only if
 the polynomial  $G$  outputted from the function {\sf
 SemiAlgebraicSolve} has only real solutions. The coordinates of
 the rational coefficients  of polynomials $f_i$ in $f= \sum_{i}  f_i^2$ have  bit length dominated by  $\tau D^{\bigO(k)}$
and the bit complexity of {\sf RationalTotalRealSolve} is
$\tau^{\bigO(1)}D^{\bigO(k)}$.

\end{theorem}

\bibliographystyle{elsart-harv} \bibliography{rationalsos}


\end{document}